\documentclass[11pt]{article}
\usepackage{fullpage,comment,algorithm,algorithmic}
\usepackage{color,colortbl}
\usepackage{float}
\usepackage{amsthm}
\usepackage{amsmath}
\usepackage{amssymb}
\usepackage{graphicx}
\usepackage[colorinlistoftodos,prependcaption,textsize=tiny]{todonotes}

\ifx\pdftexversion\undefined
\usepackage[colorlinks,linkcolor=black,filecolor=black,citecolor=black,urlco
lor=black,pdfstartview=FitH]{hyperref}
\else
\usepackage[colorlinks,linkcolor=blue,filecolor=blue,citecolor=blue,urlcolor
=blue,pdfstartview=FitH]{hyperref}
\fi

\newtheorem{theorem}{Theorem}
\newtheorem{lemma}{Lemma}
\newtheorem*{lemma*}{Lemma}

\newtheorem{claim}[lemma]{Claim}

\newtheorem{definition}{Definition}
\newtheorem{observation}[lemma]{Observation}

\newtheorem{fact}{Fact}

\newcommand{\namedref}[2]{\hyperref[#2]{#1~\ref*{#2}}}
\newcommand{\sectionref}[1]{\namedref{Section}{#1}}
\newcommand{\appendixref}[1]{\namedref{Appendix}{#1}}
\newcommand{\subsectionref}[1]{\namedref{Subsection}{#1}}
\newcommand{\theoremref}[1]{\namedref{Theorem}{#1}}
\newcommand{\defref}[1]{\namedref{Definition}{#1}}
\newcommand{\figureref}[1]{\namedref{Figure}{#1}}
\newcommand{\claimref}[1]{\namedref{Claim}{#1}}
\newcommand{\lemmaref}[1]{\namedref{Lemma}{#1}}
\newcommand{\tableref}[1]{\namedref{Table}{#1}}

\newcommand{\algref}[1]{\namedref{Algorithm}{#1}}
\newcommand{\factref}[1]{\namedref{Fact}{#1}}
\newcommand{\obsref}[1]{\namedref{Observation}{#1}}
\newcommand{\lineref}[1]{\namedref{Line}{#1}}

\newcommand{\home}{\mbox{\bf home}}

\newcommand{\res}{\mbox{\bf Res}}
\newcommand{\sur}{\mbox{\bf Sur}}

\newcommand{\eps}{\epsilon}

\newcommand{\dist}{\mbox{\bf dist}}

\newcommand{\FPSDHPP}{\mbox{\bf FPSDHPP}}

\def\beginsmall#1{\vspace{-\parskip}\begin{#1}\itemsep-\parskip}
	\def\endsmall#1{\end{#1}\vspace{-\parskip}}

\def\inline#1:{\par\vskip 7pt\noindent{\bf #1:}\hskip 10pt}

\def\inline#1:{\par\vskip 7pt\noindent{\bf #1:}\hskip 10pt}

\def\blackslug{\hbox{\hskip 1pt \vrule width 4pt height 8pt
		depth 1.5pt \hskip 1pt}}

\def\QED{\quad\blackslug\lower 8.5pt\null\par}

\newcommand{\alert}[1]{\textbf{\color{red}
[[[#1]]]}\marginpar{\textbf{\color{red}**}}\typeout{ALERT:
\the\inputlineno: #1}}

\floatstyle{ruled}
\newfloat{algorithm}{tbp}{loa}
\providecommand{\algorithmname}{Algorithm}
\floatname{algorithm}{\protect\algorithmname}

\makeatother

\usepackage{authblk}
\usepackage{pdfsync}

\newcommand{\logdiam}{\phi}

\begin{document}
\author[1]{Ittai Abraham}
\author[2]{Shiri Chechik}
\author[3]{Michael Elkin}
\author[3]{Arnold Filtser}
\author[3]{Ofer Neiman}

\affil[1]{VMWare. Email: \texttt{iabraham@vmware.com}}
\affil[2]{Tel-Aviv University. Email: \texttt{schechik@post.tau.ac.il}}
\affil[3]{Ben-Gurion University of the Negev. Email: \texttt{\{elkinm,arnoldf,neimano\}@cs.bgu.ac.il}}

\title{Ramsey Spanning Trees and their Applications}
\maketitle
\begin{abstract}
	
The \emph{metric Ramsey problem} asks for the largest subset $S$ of a metric space that can be embedded into an ultrametric (more generally into a Hilbert space) with a given distortion. Study of this problem was motivated as a non-linear version of Dvoretzky theorem. Mendel and Naor \cite{MN07} devised the so called Ramsey Partitions to address this problem, and showed the algorithmic applications of their techniques to approximate distance oracles and ranking problems.

In this paper we study the natural extension of the metric Ramsey problem to graphs, and introduce the notion of \emph{Ramsey Spanning Trees}. We ask for the largest subset $S\subseteq V$ of a given graph $G=(V,E)$, such that there exists a spanning tree of $G$ that has small stretch for $S$. Applied iteratively, this  provides a small collection of spanning trees, such that each vertex has a tree providing low stretch paths to {\em all other vertices}. The union of these trees  serves as a special type of spanner, a {\em tree-padding spanner}. We  use this spanner to devise the first compact stateless routing scheme with $O(1)$ routing decision time, and  labels which are much shorter than in all currently existing schemes.

We first revisit the metric Ramsey problem, and provide a new deterministic construction. We prove that for every $k$, any $n$-point metric space has a subset $S$ of size at least $n^{1-1/k}$ which embeds into an ultrametric with distortion $8k$. This results improves the best previous result of Mendel and Naor that obtained distortion $128k$ and required randomization. In addition, it provides the state-of-the-art deterministic construction of a distance oracle.
Building on this result, we prove that for every $k$, any $n$-vertex graph $G=(V,E)$ has a subset $S$ of size at least $n^{1-1/k}$, and a spanning tree of $G$, that has stretch $O(k \log \log n)$ between any point in $S$ and any point in $V$.
\end{abstract}

\thispagestyle{empty}
\newpage
\setcounter{page}{1}

\section{Introduction}
Inspired by the algorithmic success of Ramsey Type Theorems for metric spaces, in this paper we study an analogue Ramsey Type Theorem in a graph setting. The classical Ramsey problem for metric spaces was introduced in \cite{BFM86}, and is concerned with finding "nice" structures in arbitrary metric spaces. Following \cite{BLMN03}, \cite{MN07} showed that every $n$-point metric $(X,d)$ has a subset $M\subseteq X$ of size at least $n^{1-1/k}$ that embeds into an ultrametric (and thus also into Hilbert space) with distortion at most $O(k)$, for a parameter $k\ge 1$. In fact, they construct an ultrametric on $X$ which has $O(k)$ distortion for any pair in $M\times X$. Additionally, \cite{MN07} demonstrated the applicability of their techniques, which they denoted Ramsey Partitions, to approximate distance oracles and ranking problems.

We introduce a new notion that we call {\em Ramsey Spanning Trees}, which is a natural extension of the metric  Ramsey problem to graphs. We show that every graph $G=(V,E)$ with $n$ vertices admits a subset $M\subseteq V$ of size at least $n^{1-1/k}$, such that there exists a spanning tree of $G$ that has stretch $O(k\log\log n)$ on all pairs in $M\times V$. (The extra factor of $\log\log n$ in the stretch comes from the state-of-the-art result of $O(\log n\log\log n)$ for low stretch spanning trees \cite{AN12}. It is quite plausible that if that result is improved to the optimal $O(\log n)$, then the stretch in our result would be only $O(k)$.)

By applying this result iteratively, we can obtain a small collection of trees so that each vertex has small stretch to all other vertices in at least one of the trees. Let $\dist(u,v,G)$ denote the shortest path distance in the graph $G$ between the vertices $u,v\in V$, then our main result is the following.
\begin{theorem}\label{thm:main-col}
Let $G=(V,E)$ be a weighted graph on $n$ vertices, and fix a parameter $k\ge 1$. There is a polynomial time deterministic algorithm that finds a collection ${\cal T}$ of $k\cdot n^{1/k}$ spanning trees of $G$, and a mapping $\home:V\to{\cal T}$, such that for every $u,v\in V$ it holds that $\dist(v,u,\home(v))\le O(k\log\log n)\cdot \dist(v,u,G)$.
\end{theorem}

A spanner $H$ with stretch $t$ for a graph $G$, is a sparse spanning subgraph satisfying $\dist(v,u,H)\le t\cdot \dist(v,u,G)$.
Spanners are a fundamental metric and graph-theoretic constructions; they are very well-studied \cite{PS89,ADDJ90,C93,EP04,BS03,TZ06,AB16}, and have numerous applications \cite{Awerbuch84,ABCP93,C93,E01,GRTU16}.
\theoremref{thm:main-col} can be viewed as providing a spanner which is the union of $k\cdot n^{1/k}$ spanning trees, such that every vertex has a tree with low stretch paths to {\em all other vertices}.
We call such a spanner a {\em tree-padding spanner} of {\em order} $k \cdot n^{1/k}$. To the best of our knowledge, no previous construction of spanners can be viewed as a tree-padding spanner of order $o(n)$.
 Until now even the following weaker question was open: does there exist a spanner which is a union  of a sublinear in $n$ number of trees, such that every pair of vertices has a low stretch path in one of these trees.

Having a single tree that provides good stretch for any pair containing the vertex $v$, suggest that routing messages to or from $v$ could be done on this one tree. Our main application of Ramsey spanning trees is a compact routing scheme that has constant routing decision time and improved label size, see \sectionref{sec:apps} for more details.

\paragraph{Deterministic Ramsey Partitions.}
As a first step towards our main result, which is of interest in its own right, we provide a new deterministic Ramsey ultrametric construction. In particular, we show a polynomial time deterministic algorithm, that given an $n$-point metric space $(X,d)$ and a parameter $k\ge 1$, finds a set $M\subseteq X$ of size at least $n^{1-1/k}$ and an ultrametric $(M,\rho)$ with distortion at most $8k-2$. That is, for each $v,u\in M$,
\[
d(v,u) \le \rho(v,u) \le (8k-2) \cdot d(v,u)~.
\]
Even though our construction is deterministic, it has smaller distortion than all previous constructions.
The first result of this flavor was by Mendel and Naor \cite{MN07}, obtaining distortion of $128k$.
Belloch et. al \cite{BGS16} 
showed that the (randomized) algorithm of \cite{FRT04} constructs an ultrametric with distortion $18.5k$ (they also provided a near-linear time implementation of it).
Naor and Tao  \cite{NT12} declared that $16k$ is obtainable but that a Ramsey partition with distortion better than $16k-2$ seems not to be possible with their current techniques. Moreover, \cite{MN07} mention as a drawback that their solution is randomized (while \cite{BLMN03} is deterministic).

An application of our improved deterministic Ramsey ultrametric construction is a new distance oracle that has the best space-stretch-query time tradeoff among deterministic distance oracles. See \sectionref{sec:apps} below.

\paragraph{Techniques.}
Our construction of Ramsey ultrametrics uses the by-now-standard \emph{deterministic} ball growing approach, e.g. \cite{Awerbuch84,AKPW95,AP90,FRT04,Bartal04}. 
In this paper we provide tighter and more parameterized analysis of these multi-scale deterministic region growing techniques. Our improved analysis of the deterministic ball growing technique of \cite{FRT04,Bartal04} obtains a similar type of improvement as the one obtained by the analysis of Mendel and Naor \cite{MN07} on the randomized partition technique of \cite{CKR01,FRT}.

Our construction of Ramsey spanning trees is based on combining ideas from our Ramsey ultrametric construction, with the Petal Decomposition framework of \cite{AN12}. The optimal multi-scale partitions of \cite{FRT04,Bartal04} cannot be used in this petal decomposition framework, so we must revert to partitions based on \cite{S95,EEST05}, which induce an additional factor of $O(\log\log n)$ to the stretch. In addition, the refined properties required by the Ramsey partition make it very sensitive to constant factors (these constants can be ignored in the \cite{EEST05} analysis of the average stretch, say).
In order to alleviate this issue, we consider two possible region growing schemes, and choose between them according to the densities of points that can still be included in $M$. One of these schemes is a standard one, while the other grows the region "backwards", in a sense that it charges the remaining graph, rather than the cluster being created, for the cost of making a cut. See \sectionref{subsec:createPetal} for more details.

\subsection{Applications}\label{sec:apps}

\paragraph{Distance Oracles.} A distance oracle is a succinct data structure that (approximately) answers distance queries. A landmark result of \cite{TZ01} states that any metric (or graph) with $n$ points has a distance oracle of size $O(k\cdot n^{1+1/k})$,\footnote{We measure size in machine words, each words is $\Theta(\log n)$ bits.} that can report any distance in $O(k)$ time with stretch at most $2k-1$. A deterministic variant with the same parameters was given by \cite{RTZ05}, and this was the state-of-the-art for deterministic constructions. The oracle of \cite{MN07} has improved size $O(n^{1+1/k})$ and $O(1)$ query time, but larger stretch $128k$. This oracle was key for subsequent improvements by \cite{W13,C14,C15}, the latter gave a randomized construction of an oracle with size $O(n^{1+1/k})$, query time $O(1)$ and stretch $2k-1$ (which is asymptotically optimal assuming Erdos' girth conjecture).

Similarly to \cite{MN07}, our deterministic construction of Ramsey ultrametrics can provide a deterministic construction of an approximate distance oracle.

\begin{theorem}\label{thm:small-constnat}
For any metric space on $n$ points, and any $k>1$, $0<\epsilon<1$, there is an efficient deterministic construction of a distance oracle of size $O(n^{1+1/k})$, that has stretch $8 (1+\eps)k$ and query time $O(1/\eps)$.
\end{theorem}

This is the first deterministic construction of an approximate distance oracle with constant query time and small size $O(n^{1+1/k})$.

Moreover, our oracle is an essential ingredient towards de-randomizing the recent distance oracles improvements \cite{W13,C14,C15}.
Specifically, if we construct \cite{C14} by replacing
the distance oracle of Mendel and Naor \cite{MN07} by our deterministic version, and replacing the distance oracle of
Thorup and Zwick \cite{TZ01} by the deterministic version of Roditty, Thorup, and Zwick \cite{RTZ05},
we immediately get a deterministic distance oracle of $O(k \cdot n^{1+1/k})$ size, $2k-1$ stretch and $O(1)$ query time.
This is a strict improvement over  \cite{RTZ05}.
In addition, our oracle can be viewed as a first step towards de-randomizing the \cite{C15} oracle.
A summary of all the previous and current results can be found at \tableref{tab:DistanceOracle} in \appendixref{sec:table}.

\paragraph{Routing with Short Labels and Constant Decision Time.}
A routing scheme in a network is a mechanism that allows packets to be delivered from any node to any other node. The network is represented as a weighted undirected graph, and each node can forward incoming data by using local information stored at the node, often called a routing table, and the (short) packet's header. The routing scheme has two main phases: in the preprocessing phase, each node is assigned a routing table and a short label. In the routing phase, each node receiving a packet should make a local decision, based on its own routing table and the packet's header (which may contain the label of the destination, or a part of it), where to send the packet.
The {\em routing decision time} is the time required for a node to make this local decision.
The {\em stretch} of a routing scheme is the worst ratio between the length of a path on which a packet is routed, to the shortest possible path. A routing scheme is called {\em stateless} if the routing decision does not depend on the path traversed so far.

The classical routing scheme of \cite{TZ01b}, for a graph on $n$ vertices and integer parameters $k,b> 1$, provides a scheme with routing tables of size $O(k\cdot b\cdot n^{1/k})$, labels of size $(1+o(1))k\log_bn$, stretch $4k-5$, and decision time $O(1)$ (but the initial decision time is $O(k)$). The stretch was improved recently to roughly $3.68k$ by \cite{C13}, using a similar scheme as \cite{TZ01b}. With \theoremref{thm:main-col}, we devise a stateless compact routing scheme with very short labels, of size only $(1+o(1))\log_bn$, and with {\em constant} decision time, while the stretch increases to $O(k\log\log n)$ (and with the same table size as \cite{TZ01b}).

We wish to point out that our construction of a routing scheme is simpler in some sense that those of \cite{TZ01b,C13}. In both constructions there is a collection of trees built in the preprocessing phase, such that every pair of vertices has a tree that guarantees small stretch. Routing is then done in that tree. In our construction there are few trees, so every vertex can store information about all of them, and in addition, every vertex $v\in V$ knows its home tree, and routing towards $v$ from {\em any other vertex} on the tree $\home(v)$ has small stretch. In particular, the header in our construction consists of only the label of the destination. In the \cite{TZ01b} scheme, however, there are $n$ trees, and a certain process is used to find the appropriate tree to route on, which increases the initial decision time, and also some information must be added to the header of the message after the tree is found. Finally, our routing scheme is stateless, as opposed to \cite{TZ01b}. (We remark that using ideas from \cite{C14}, one can devise a stateless routing scheme based on \cite{TZ01b}, but this scheme seems to suffer from larger header and decision time at each node.)

\begin{theorem}\label{thm:route}
	Given a weighted graph $G=(V,E)$ on $n$ vertices and integer parameters $k,b> 1$, there is a stateless routing scheme with stretch $O(k\log\log n)$ that has routing tables of size $O(k\cdot b\cdot n^{1/k})$ and labels of size $(1+o(1))\log_bn$. The decision time in each vertex is $O(1)$.
\end{theorem}
Observe that choosing parameters $2k$ and $b=n^{1/(2k)}$ for \theoremref{thm:route} yields a routing scheme with stretch $O(k\log\log n)$ that has tables of size $O(k\cdot n^{1/k})$ and labels of size only $O(k)$. Another interesting choice of parameters is $b=2$ and $k=\frac{100\log n}{\log\log n}$, this provides a scheme with stretch $O(\log n)$ that has tables of size $O(\log^{1.01}n)$ and labels of size $O(\log n)$. Compare this to the \cite{TZ01b} scheme, which for stretch $O(\log n)$ has tables of size $O(\log n)$ and labels of size $O(\log^2 n)$.

\subsection{Organization}
In \sectionref{sec:Ramsey-Partitions} we present our deterministic Ramsey partitions, that are used for Ramsey ultrametrics and distance oracles. In \sectionref{sec:Ramsey-spanning} we show the Ramsey spanning trees, and the application to routing. Each section can be read independently.

\section{Preliminaries}
Let $G=(V,E)$ be a weighted undirected graph. We assume that the minimal weight of an edge is $1$. For any $Y\subseteq V$ and $x,y\in Y$, denote by $\dist(x,y,Y)$ the shortest path distance in $G[Y]$ (the graph induced on $Y$).
For $v\in Y$ and $r\ge 0$ let
$B(v,r,Y)=\{u\in Y\mid \dist(v,u,Y)\le r\}$, when $Y=V$ we simply write $B(v,r)$. We may sometimes abuse notation and not distinguish between a set of vertices and the graph induced by them.

An ultrametric $\left(Z,d\right)$ is a metric space satisfying a
strong form of the triangle inequality, that is, for all $x,y,z\in Z$,
$d(x,z)\le\max\left\{ d(x,y),d(y,z)\right\} $. The following definition
is known to be an equivalent one (see \cite{BLMN05}).
\begin{definition}\label{def:ultra}
	An ultrametric is a metric space $\left(Z,d\right)$ whose elements
	are the leaves of a rooted labeled tree $T$. Each $z\in T$ is associated
	with a label $\ell\left(z\right)\ge0$ such that if $q\in T$ is a
	descendant of $z$ then $\ell\left(q\right)\le\ell\left(z\right)$
	and $\ell\left(q\right)=0$ iff $q$ is a leaf. The distance between
	leaves $z,q\in Z$ is defined as $d_{T}(z,q)=\ell\left(\mbox{lca}\left(z,q\right)\right)$
	where $\mbox{lca}\left(z,q\right)$ is the least common ancestor of
	$z$ and $q$ in $T$.
\end{definition}

\section{Ramsey Partitions}\label{sec:Ramsey-Partitions}

Consider an undirected weighted graph $G=(V,E)$, and a parameter $k\ge 1$.
Let $D$ be the diameter of the graph and let $\logdiam = \lceil \log{(D+1)} \rceil$.
Let $\rho_i = 2^i/(4k)$.
We start by presenting a construction for a collection ${\cal S}$ of cluster partial partitions ${\cal X}_i$ satisfying the following key properties.

\begin{definition}\label{def:FPSDC}[$(G,U,k)$-Fully Padded Strong Diameter Hierarchical Partial Partition]
	Given a graph $G=(V,E)$, an index $k$ and a set of nodes $U\subseteq V$, a $(G,U,k)$-Fully Padded Strong Diameter Hierarchical Partial Partition (\FPSDHPP) is a collection ${\cal X}_i$
	of subsets of nodes $X \in {\cal X}_i$ with a center $r(X)$ for $0 \leq i \leq \logdiam$ with the following properties.

	\begin{description}
		\item{(i)}
		For every $0 \leq i \leq \logdiam$,
		the subsets in ${\cal X}_i$ are disjoint, namely, for every two different subsets $X,X' \in {\cal X}_i$, $X \cap X' = \emptyset$.
		\item{(ii)}
		For every $0 \leq i < \logdiam$ and
		every subset $X \in {\cal X}_i$, there exists a subset $X' \in {\cal X}_{i+1}$ such that $X \subseteq X'$.
		\item{(iii)}
		For every $0 \leq i \leq \logdiam$ and
		every $X\in {\cal X}_i$ and every $v\in X$, $\dist(v,r(X),X) < 2^i$.
		\item{(iv)}
		There exists a set $\hat{V} \subseteq U$ such that $|\hat{V}| \geq |U|^{1-1/k}$ and
		for every $v \in \hat{V}$ and every $i$, there exists a subset $X\in {\cal X}_i$
		such that $B(v,\rho_i) \subseteq X$.
	\end{description}
	
\end{definition}

For a node $v$ and index $i$, we say that $v$ is $i$-padded in ${\cal S}$, if there exists a subset $X \in {\cal X}_i$ such that
$B(v,\rho_i) \subseteq X$. We would like to maximize the number of nodes that are padded on all levels.

\paragraph{Fully Padded Strong Diameter Hierarchical Partial Partition Construction:}

Let us now turn to the construction of the collection ${\cal S}$ of cluster partial partitions ${\cal X}_i$ given a set $U$.

In the beginning of the algorithm, all nodes in $U$ are set as \emph{marked}. The algorithm iteratively \emph{unmarks} some of the nodes.
The nodes that will remain marked by the end of the process are the nodes that are padded on all levels.
For a given graph $H$, let $B_M(v,d,H)$ ($M$ stands for marked) be the set of marked nodes at distance at most $d$ from $v$ in $H$.

For a subgraph $G'$ and a node $v \in V(G')$,
let $Z_i(v,G') = |B_M(v,2^i,G')|/|B_M(v,2^{i-1},G')|$.
The construction given in \algref{alg-strong-ramsey}.

\begin{algorithm}[h]
	\caption{$
		\mathcal{S}=\FPSDHPP(G,U)$}
	\begin{algorithmic}[1]\label{alg-strong-ramsey}
		\STATE Mark all the nodes in $U$.
		\STATE  Set ${\cal X}_{\logdiam}=\{V\}$ to be the trivial partition. Set $r(V)\in V$ to be the vertex $v$ with maximal $|B_M(v,2^{\phi-1},G)|$.
		\FOR {$i$ from ${\logdiam-1}$ to $0$}
		\FOR {every subset $X \in {\cal X}_{i+1}$}\label{line:pick_Subset}
		\STATE Set $H$ to be the induced graph on $X$.
		\WHILE {$H$ contains a marked vertex}
		\STATE Pick a node $v \in V(H)$ with maximal $|B_M(v,2^{i-1},H)|$.
		\STATE Let $j(v)\ge 0$ be the minimal integer such that\\ \hspace{40pt}$|B_M(v,2^{i-1} + 2(j(v)+1) \rho_i, H)| \leq
		|B_M(v,2^{i-1} + 2j(v) \rho_i, H)| \cdot |Z_i(v,H)|^{1/k}$.
		\label{line:RG}
		\STATE Let $X(v)=B\left(v,2^{i-1}+(2j(v)+1)\rho_{i},H\right)$.
		\label{line:X(v)def}
		\STATE Add $X(v)$ to  ${\cal X}_i$.
		\STATE Unmark the nodes in $B(v,2^{i-1} + 2(j(v)+1) \rho_i, H) \setminus B(v,2^{i-1} + 2j(v) \rho_i, H)$.
		\label{line:unMark}
		\STATE Remove all nodes in $X(v)$ from $H$. \label{line:RemoveX(v)}
		\ENDWHILE
		\ENDFOR	
		\ENDFOR
		\STATE \textbf{set} $\hat{V}$ to be all nodes that remain marked.
	\end{algorithmic}
\end{algorithm}

Let $X(v)$ be a set constructed in \lineref{line:X(v)def} of \algref{alg-strong-ramsey}, when partitioning $X\in{\cal X}_{i+1}$.
We say that $X$ is the parent of $X(v)$.
Let $H(X(v))$ denote the graph in \algref{alg-strong-ramsey} just before $X(v)$ was constructed (note that this is a graph induced on a subset of $X$).
We say that $B(v,2^{i-1} + 2j(v) \rho_i, H(X(v)))$ is the {\em interior} part of $X(v)$
.
We also say that the set $B_M(v,2^{i-1} + 2(j(v)+1) \rho_i, H(X(v)))$ is the {\em responsibility set} of $X(v)$, hereafter referred to as $\res(X(v))$.
Note that every node $u$ that is still marked after the processing of $X$ is completed, belongs to exactly one set $\res(X(v))$ for $X(v) \in {\cal X}_i$.

We now define by induction the term {\em $i$-surviving} for $0\leq i \leq \logdiam$:
All nodes in $U$ are ${\logdiam}$-surviving.
We say that a node is $i$-surviving if it is $(i+1)$-surviving and it belongs to the interior part of some subset in ${\cal X}_{i}$.
Our goal in the analysis is to show that many nodes are $0$-surviving, which is exactly the set ${\hat{V}}$.
For a subset $X \in {\cal X}_{i}$, let $\sur(X)$ be the set of nodes in $X$ that are $0$-surviving.
We now turn to the analysis.

The next auxiliary claim helps in showing that property $(iii)$ holds.
\begin{claim}
	\label{claim:j-v}
	Consider a subset $X \in {\cal X}_i$ centered at some node $v = r(X)$.
	The index $j(v)$ defined in \lineref{line:RG} of \algref{alg-strong-ramsey}, satisfies $j(v) \leq k-1$
\end{claim}
\begin{proof}
	Seeking contradiction, assume that for every $0\le j'\le k-1$, $|B_M(v,2^{i-1} + 2(j'+1) \rho_i, H(X))| > |B_M(v,2^{i-1} + 2j' \rho_i, H(X))| \cdot |Z_i(v,H(X))|^{1/k}$. Then applying this for $j'=k-1,k-2,\dots,0$ we get
	\begin{align*}
	|B_{M}(v,2^{i},H(X))| & =|B_{M}(v,2^{i-1}+2k\rho_{i},H(X))|\\
	& >|B_{M}(v,2^{i-1}+2(k-1)\rho_{i},H(X))|\cdot|Z_{i}(v,H(X))|^{1/k}\\
	&  > \dots>|B_{M}(v,2^{i-1},H(X))|\cdot|Z_{i}(v,H(X))|^{k/k}\\
	& =|B_{M}(v,2^{i},H(X))|~,
	\end{align*}
	a contradiction.
\end{proof}

The next lemma shows that the collection ${\cal S}$ satisfies properties $(i)-(iii)$.

\begin{lemma}
	\label{lem:tree-cover-3}
	${\cal S}$ satisfy properties $(i)-(iii)$.
\end{lemma}
\begin{proof}
	Property $(i)$ is straightforward from \lineref{line:RemoveX(v)}.
	Property $(ii)$ holds as each $X(v)$ is selected from the graph $H(X(v))$, which is an induced graph over a subset of $X$ (the parent of $X(v)$). Finally, property $(iii)$ follows from
	\claimref{claim:j-v}, as the radius of $X(v)$ is bounded by $2^{i-1}+(2(k-1)+1)\rho_{i}<2^{i-1}+2k\rho_{i}=2^{i}$.
\end{proof}

Next we argue that if a vertex is $0$-surviving, then it is padded in all the levels.
\begin{lemma}\label{lem:SurPaddedAll}
	Suppose $x\in \sur(V)$, then $x$ is padded in all the levels.
\end{lemma}
\begin{proof}
	Fix some $x\in \sur(V)$. To prove that $x$ is $i$-padded, we assume inductively that $x$ is $j$-padded for all $i< j\le\logdiam$ (the base case $i=\logdiam$ follows as $B(x,\rho_\logdiam)\subseteq V$).
	Let $X\in\mathcal{X}_{i+1}$ such that $x\in X$. Set $B=B(x,\rho_{i})$. By the induction hypothesis $B\subseteq B(x,\rho_{i+1})\subseteq X$.
	Let $X(v)\in \mathcal{X}_{i}$ such that $x\in X(v)$.
	
First we argue that $B\subseteq H(X(v))$. Seeking contradiction, let $X(v')\in \mathcal{X}_{i}$ be the first created cluster such that there is $u\in B\cap X(v')$.
	By the minimality of $v'$, it follows that $B\subseteq H(X(v'))$. Thus $\dist(v,u,H(X(v')))=\dist(v,u,G)\le \rho_{i}$.
	Let $j(v')$ the index chosen in \lineref{line:RG} of \algref{alg-strong-ramsey}. Then $X(v')=B\left(v',2^{i-1}+(2j(v')+1)\rho_{i},H(X(v'))\right)$. Using the triangle inequality $\dist(x,v',H(X(v')))\le  2^{i-1}+(2j(v')+2)\rho_{i}$. Therefore $x$ was unmarked in \lineref{line:unMark}, a contradiction.
	
	It remains to show that $B\subseteq X(v)$. Set $j(v)$ s.t. $X(v)=B\left(v,2^{i-1}+(2j(v)+1)\rho_{i},H(X(v))\right)$. As $x$ is part of the interior of $X(v)$, it holds that
	$\dist(x,v,H(X(v)))\le 2^{i-1}+2j(v)\rho_{i}$. Therefore $B\subseteq B\left(v,2^{i-1}+(2j(v)+1)\rho_{i},H(X(v))\right)=X(v)$.
\end{proof}

The next lemma bounds the number of surviving nodes.
\begin{lemma}\label{lem:frac-survival}
	For every index $i$ and every subset $X = X(v) \in {\cal X}_i$, the 0-surviving nodes satisfy
	$|\sur(X)|\ge\left|M_i(X)\right|/\left|B_M(r(X),2^{i-1},H(X))\right|^{\frac{1}{k}}$
\end{lemma}
\begin{proof}
	We prove the lemma by induction on $i$.
	Consider first the base case where $X = X(v) \in {\cal X}_0$.
	As the subsets in ${\cal X}_0$ contain a single node (their radius is less than 1), it holds that	
	$|\sur(X)|=1= \frac{1}{1}=\left|M_0(X)\right|/\left|B_M(r(X),1/2,H(X))\right|^{\frac{1}{k}}$ (observe that each cluster in ${\cal X}_i$ has at least 1 marked node, for all $0\le i\le\logdiam$).
	Assume the claim holds for every subset $X' \in {\cal X}_i$, and consider $X \in {\cal X}_{i+1}$.
	Let $v =r(X)$.
	Consider the children $X_1,...,X_{j'}$ of $X$.
	For every $1\leq h \leq j'$, set $v_h=r(X_h)$.
	
	Note that by definition of $j(v_h)$ in \lineref{line:RG} of \algref{alg-strong-ramsey} and by the construction of $X_h$ in \lineref{line:X(v)def}, we have
	that $|M_i(X_h)|\ge |\res(X_h)|/Z_{i}(v_h,H(X_h))^{1/k}$.
	Moreover, by the induction hypothesis we have that
	$|\sur(X_h)|\ge\left|M_i(X_h)\right|/\left|B_M(v_h,2^{i-1},H(X_h))\right|^{\frac{1}{k}}$ for every $h$.

	We claim that $|B_M(v,2^{i},H(X))|  \geq |B_M(v_h,2^{i},H(X_h))|$ for every $1\leq h \leq j'$.
	To see this, note that $v$ is the node with maximal $|B_M(v,2^{i},H(X))|$, hence $|B_M(v,2^{i},H(X))| \geq |B_M(v_h,2^{i},H(X))|$.
	In addition, note that $H(X_h) \subseteq H(X)$, hence
	$
	|B_M(v_h,2^{i},H(X))| \geq  |B_M(v_h,2^{i},H(X_h))|$.
	It follows that $|B_M(v,2^{i},H(X))|  \geq |B_M(v_h,2^{i},H(X_h))|$.
	
	Therefore the number of 0-surviving nodes in $X_h$ is at least
	
	\begin{eqnarray*}
		|\sur(X_h)|
		&\ge&\frac{\left|M_i(X_h)\right|}{\left|B_M(v_h,2^{i-1},H(X_h))\right|^{\frac{1}{k}}}
		\ge
		 \frac{\left|\res(X_h)\right|/\left|Z_{i}(v_h,H(X_h))\right|^{\frac{1}{k}}}{\left|B_M(v_h,2^{i-1},H(X_h))\right|^{\frac{1}{k}}}\\
		&=& 	 \frac{	\left|\res(X_h)\right|}{\left|B_M(v_h,2^{i},H(X_h))\right|^{\frac{1}{k}}}
		\cdot \frac{\left|B_M(v_h,2^{i-1},H(X_h))\right|^{\frac{1}{k}}}{\left|B_M(v_h,2^{i-1},H(X_h))\right|^{\frac{1}{k}}}\\
		&\ge& 	 \frac{	\left|\res(X_h)\right|}{\left|B_M(v,2^{i},H(X))\right|^{\frac{1}{k}}}~,\\
	\end{eqnarray*}
	we conclude that
	\[
	|\sur(X)|=\sum_{h=1}^{j'}|\sur(X_h)|\ge \sum_{h=1}^{j'}\frac{	 \left|\res(X_h)\right|}{\left|B_M(v,2^{i},H(X))\right|^{\frac{1}{k}}}=\frac{	 \left|M_{i+1}(X)\right|}{\left|B_M(v,2^{i},H(X))\right|^{\frac{1}{k}}}~.\qedhere
	\]
\end{proof}

Using \lemmaref{lem:frac-survival} on $V$ with $i=\logdiam$, combined with \lemmaref{lem:SurPaddedAll}, implies property $(iv)$.

\begin{lemma}
	The number of marked nodes  $\hat{V}$ by the end of \algref{alg-strong-ramsey} is at least $|U|^{1-1/k}$.
	Moreover, for every $v \in \hat{V}$ and every $i$, there exists a subset $X\in {\cal X}_i$
	such that $B(v,\rho_i) \subseteq X$.
\end{lemma}

\begin{theorem}\label{thm:DRP}
	For every $n$-point metric space and $k\ge 1$, there exists a subset of size $n^{1-1/k}$ that can be embedded into an ultrametric with distortion $8k-2$.
\end{theorem}
\begin{proof}
The hierarchical partial partition ${\cal S}=\{{\cal X}_i\}$ naturally induce an ultrametric on $\hat{V}$. The singleton sets of $\hat{V}$ are the leaves, and each $X\in{\cal X}_i$ for $0\le i<\logdiam$ will be a tree-node which is connected to its parent. Each set in ${\cal X}_i$ for $i\ge 1$ will receive the label $2^{i+1}(1-1/(4k))$, while the leaves in ${\cal X}_0$ receive the label $0$ (recall \defref{def:ultra}).

	Consider two nodes $u,v \in \hat{V}$. Assume the least common ancestor of $u,v$ is $X\in{\cal X}_i$, for some $1\le i\le\logdiam$.
	Hence $\dist(u,v,G) \le 2 \cdot (2^i -  2^i/4k)$ (they are both in the interior of $X$ - a ball with radius $\le 2^i-2\rho_i$). Since this is the label of $X$, we conclude that distances in the ultrametric are no smaller than those in $G$.
	
	Next we argue that distances increase by a factor of at most $8k-2$. Consider any $u,v$ as above, and seeking a contradiction, assume that $\dist(v,u,G)<\frac{2^{i+1}(1-1/(4k))}{8k-2} = \rho_i$. Let $P$ be the shortest path from $v$ to $u$ in $G$.	As $v$ was padded in $X$, necessarily $P\subseteq X$. Consider the first time a vertex $z\in P$ was added to a cluster $X'\in\mathcal{X}_{i-1}$, then $P\subseteq H(X')$. Let $j$ be such that $X'= B(r(X'),2^{i-2} + (2j+1)\rho_{i-1} , H(X'))$.
Since $P$ is a shortest path, at least one of $u,v$ must be within distance less than $\rho_i/2=\rho_{i-1}$ from $z$, w.l.o.g assume $\dist(v,z,H(X'))\le \rho_{i-1}$. This implies that $v\in\res(X')=B(r(X'),2^{i-2} + (2j+2)\rho_{i-1} , H(X'))$, and as $v$ is marked it must lie in the interior of $X'$, which is $B(r(X'),2^{i-2} + 2j\rho_{i-1} , H(X'))$. But then the triangle inequality yields that $u\in\res(X')\setminus X'$, which is a contradiction to the fact that $u\in\hat{V}$.
\end{proof}

\subsection{Distance Oracle}\label{sec:Distance-Oracle}
We show a distance oracle with $O(n^{1+1/k})$ size, $(8+\eps)k$ worst case stretch and $O(1/\eps)$ query time (which is $O(1)$ for any fixed epsilon).

For simplicity we start by showing a construction with $O(k \cdot n^{1+1/k})$ size, $16k$ stretch and $O(1)$ query time. We will later see how to reduce the size and stretch.
Let $D$ be the diameter of the graph.

Our distance oracle is constructed as follows.
\beginsmall{enumerate}
\item
Set $U \gets V$.
\item
Construct the collection of cluster-partial-partitions ${\cal S}(U)$ on the graph $G$ and the set $U$.
Remove from $U$ the set of nodes $\hat{V}$ that were padded in all levels in ${\cal S}(U)$.
Continue this process as long as $U \neq \emptyset$.
\item
Let ${\cal M}$ be the set of all collections ${\cal S}(U)$ that were constructed by this process.
\item
For every ${\cal S} \in {\cal M}$ construct a cluster $X({\cal S})$ as follows.
\item
Let ${\cal S} = \{{\cal X}_0,...,{\cal X}_{\logdiam}\}$.
All nodes $V$ are the leaves (recall that only nodes in $\hat{V}$ are in ${\cal X}_0$).
For every index $i$ and every set $X \in {\cal X}_i$, add an intermediate node. Connect $X$ to its parent set.
Connect every node $v \in V$ to the set $X \in {\cal X}_i$ of minimal $i$ such that $v \in X$.
This completes the construction of $X({\cal S})$.
\item
In addition, we preprocess $X({\cal S})$ so that least common ancestor (LCA) queries could be done in constant time.
In order to do that we invoke any scheme that takes a tree and preprocess it in linear time so that LCA queries can be answered in constant time (see \cite{HaTa84,BeFa00}).
\item
Finally, note that for every node $v$ there exists a collection ${\cal S} \in {\cal M}$, where $v$ is padded in all levels.
Denote this collection by $\home(v)$.
\endsmall{enumerate}

The query phase is done as follows.
Given two nodes $s$ and $t$.
Let ${\cal S} = \home(s)$ and let ${\cal S} = \{{\cal X}_i \mid 1\leq i \leq \logdiam\}$.
Find the least common ancestor of $s$ and $t$ in $X({\cal S})$ and let $i$ be its level.
Namely, let $\mu \in X({\cal S})$ be the least common ancestor of $s$ and $t$ and let $X$ be the cluster $\mu$ represents,
the index $i$ is the index such that $X \in {\cal X}_i$.
Return $2^{i+1}$ (denoted by $\hat{\dist}(s,t)$).

\begin{lemma}
	$\dist(s,t) \leq \hat{\dist}(s,t) < 16k\cdot\dist(s,t)$.
\end{lemma}
\begin{proof}
	Let $d = \dist(s,t,G)$ and let $j$ be the  index such that $2^{j-1}<d \leq 2^{j}$.
	Let $X_i \in {\cal X}_i \in {\cal S}=\home(s)$ be the $i$ level subset such that $s \in X_i$.
	Recall that $s$ is padded in all the subsets $X_i$ for $0\leq i \leq \logdiam$.
	
	Note that $X_i$ has diameter smaller than $2\cdot 2^i$ (follows from property $(iii)$). Therefore $t\in X_i$ implies that $\dist(s,t,G)<2\cdot2^{i}=2^{i+1}$. In particular, $t \notin X_i$ for every $i < j-1$.
	Hence the least common ancestor is at least at level $j-1$. Hence the minimal distance returned by the algorithm is  $\hat{\dist}(s,t) \geq 2^j \geq d$.
	
	It remains to show that $\hat{\dist}(s,t) \leq 16k\cdot d$.
	Let $i$ be the level of $s$ and $t$'s least common ancestor.
	Note that $t \notin X_{i-1}$.
	Also recall that $s$ is padded in $X_{i-1}$ and thus $B(s,\rho_{i-1}) \subseteq X_{i-1}$, which implies $d\geq \rho_{i-1} = 2^{i-1}/(4k) = \hat{\dist}(s,t)/(16k)$.
\end{proof}

Let $U_i$ be the set $U$ after constructing the first $i$ collections.
Note that $|U_{i+1}| \leq |U_i| - |U_i|^{1-1/k}$.
By resolving this recurrence relation one can show that the number of phases is $O(k n^{1/k})$ (see \cite[Lemma 4.2]{MN07})
.
Notice that for every ${\cal S} \in {\cal M}$, $T({\cal S})$ is of size $O(n)$.
Hence, the size of our data structure is $O(k n^{1+1/k})$.

\paragraph{Reducing the size of the data structure:}
We now show how to reduce the size of the data structure to $O(n^{1+1/k})$.
We only outline the modifications to the algorithm and the analysis and omit the full details.

Here we will use only the metric structure of the graph $G$, while ignoring the structure induced by the edges.
Specifically, in line (2) of the algorithm, instead of the graph $G$, we will use the graph $G_U$ which is the complete graph over $U$, where the weight of each edge $\{u,v\}$ is equal to $\dist(u,v,G)$.
This change allows us to remove the nodes $\hat{V}$ from $G_U$ after each iteration.

The query algorithm, given two nodes $s$ and $t$ is as follows.
Let ${\cal S}_s = \home(s)$ and ${\cal S}_t = \home(t)$, and assume w.l.o.g. that
${\cal S}_s$ was constructed before ${\cal S}_t$.
Find the least common ancestor of $s$ and $t$ in $X ({\cal S}_s)$ and let $i$ be its level.
Return $2^{i+1}$.

Following the analysis of the previous construction we can show that properties $(i)-(iv)$ are satisfied and that the stretch is bound by $16k$.
The size of the data structure is bounded by $O(n^{1+1/k})$ (see \cite{MN07}, Lemma 4.2.

\paragraph{Reducing the stretch to $8(1+\eps)k$:}
We now explain how to reduce the stretch to $8(1+\eps)k$.
Note that we lose a factor of 2 in the stretch since we look on distances in multiplies of two.
Recall that in the algorithm, for a pair of vertices $s,t$ at distance $d$, we looked on the minimal index $j$ such that $d \leq 2^{j}$.
It may happen that $d$ is only slightly larger than $2^{j-1}$.
Note that by just considering all distances $(1+\eps)^i$ rather than all distances $2^i$, we get that the number of nodes that are padded in all levels is a fraction of
$1/n^{1/(\eps k)}$ rather than $1/n^{1/k}$, which is dissatisfying.
So instead we construct $O(1/\eps)$ different copies of our data structure, one for each $1+ \ell \eps$ for $0\leq \ell < 1/\eps$.
In the copy $\ell$ of the data structure we consider distances $(1+\ell \eps)2^i$ for every $0\leq i \leq \logdiam$.
Specifically, $i$-clusters have radius bounded by $(1+\ell \eps)2^i$, while the padding parameter is $\rho_{\ell,i}=(1+\ell \eps)\rho_i$.
We denote by $\home_\ell(s)$ the collection $\mathcal{S}$, created for the $\ell$'s distance oracle, where $s$ is padded in all levels.
The distance estimation of the $\ell$'s copy (denoted by $\hat{\dist}_\ell(s,t)$),
will be $(1+(\ell+1)\eps)2^{i_\ell}$, where $i_\ell$ is the level of the least common ancestor of $s$ and $t$ in $\home_\ell(s)$.

Set $d=\dist(s,t)$. For every $\ell$, we have
\begin{equation}\label{eq:DO-UB}
d>\rho_{\ell,i_\ell-1}=\frac{(1+\ell\eps)2^{i_\ell-1}}{4k}=\frac{(1+\ell\eps)}{(1+(\ell+1)\eps)}\cdot\frac{\hat{\dist}_{\ell}(s,t)}{8k}\ge\frac{\hat{\dist}_{\ell}(s,t)}{8(1+\eps)k}~.\end{equation}
From the other hand, there exist indices $\ell',j$ such that  $(1+\ell' \eps)2^{j-1} < d\le (1+(\ell'+1) \eps)2^{j-1}$. Following the analysis above, as $t$ does not separated from $s$ at level ${i_{\ell'}}$, it holds that ${i_{\ell'}}\ge j-1$.  Therefore
\begin{equation}\label{eq:DO-LB}
\hat{\dist}_{\ell'}(s,t)=(1+(\ell+1)\eps)2^{i_{\ell'}}\ge(1+(\ell+1)\eps)2^{j-1}\ge d~.
\end{equation}
In the query phase we iterate over all $O(1/\eps)$ copies, invoke the query algorithm in each copy and return the largest distance.
By equations (\ref{eq:DO-UB}) and (\ref{eq:DO-LB}), the stretch is $8(1+\eps)k$ rather than $16k$.
The query time is $O(1/\eps)$ which is $O(1)$ for every fixed $\eps$.

\section{Ramsey Spanning Trees}\label{sec:Ramsey-spanning}

In this section we describe the construction of Ramsey spanning trees, each tree will be built using the petal decomposition framework of \cite{AN12}.
Roughly speaking, the petal decomposition is an iterative method to build a spanning tree of a given graph. In each level, the current graph is partitioned into smaller diameter pieces, called petals, and a single central piece, which are then connected by edges in a tree structure. Each of the petals is a ball in a certain metric. The main advantage of this framework is that it produces a spanning tree whose diameter is proportional to the diameter of the graph, while allowing large freedom for the choice of radii of the petals. Specifically, if the graph diameter is $\Delta$, the spanning tree diameter will be $O(\Delta)$, and each radius can be chosen in an interval of length $\approx\Delta$. For the specific choice of radii that will ensure a sufficient number of vertices are fully padded, we use a region growing technique based on ideas from \cite{S95,EEST05}.

\subsection{Preliminaries}\label{sec:prel}

For subset $S\subseteq G$ and a root vertex $r\in S$, the radius of $S$ w.r.t $r$,  $\Delta_r(S)$, is the minimal $\Delta$ such that $B(r,\Delta,S)=S$.
	(If for every $\Delta$, $B(r,\Delta,S)\neq S$, (this can happen iff $S$ is not connected) we say that $\Delta_r(S)=\infty$.)
When the root $r$ is clear from context or is not relevant, we will omit it.
\begin{definition}
	[Strong Diameter Hierarchical Partition]
	Given a graph $G=(V,E)$, a Strong Diameter Hierarchical Partition (\textbf{SDHP}) is a collection $\{\mathcal{A}_i\}_{i\in[\Phi]}$ of partitions of $V$, where each cluster in each partition has a root, such that:
	\begin{itemize}
		\item $\mathcal{A}_\Phi=\{V\}$ (i.e., the first partition is the trivial one).
        \item $\mathcal{A}_1=\{\{v\}_{v\in V}\}$ (i.e., in the last partition every cluster is a singleton).
		\item  For every $1\le i<\Phi$ and $A\in \mathcal{A}_i$, there is $A'\in \mathcal{A}_{i+1}$ such that $A\subseteq A'$ (i.e., $\mathcal{A}_i$ is a refinement of $\mathcal{A}_{i+1}$).
		Moreover, $\Delta(A)\le \Delta(A')$.
		
	\end{itemize}
\end{definition}
\begin{definition}[Padded, Fully Padded]
	Given a graph $G=(V,E)$ and a subset $A\subseteq V$, we say that a vertex $y\in A$ is {\em $\rho$-padded} by a subset $A'\subseteq A$ (w.r.t $A$) if $B(y,\Delta(A)/\rho,G)\subseteq A'$. See \figureref{fig:paddedVertex} for illustration.\\
	We say that $x\in V$ is {\em $\rho$-fully-padded} in the SDHP $\{\mathcal{A}_i\}_{i\in[\Phi]}$, if for every $2\le i\le \Phi$ and $A\in \mathcal{A}_i$ such that $x\in A$, there exists $A'\in \mathcal{A}_{i-1}$ such that $x$ is $\rho$-padded by $A'$ (w.r.t $A$).
\end{definition}
\begin{figure}
	\begin{center}
		\includegraphics[width=0.41\textwidth]{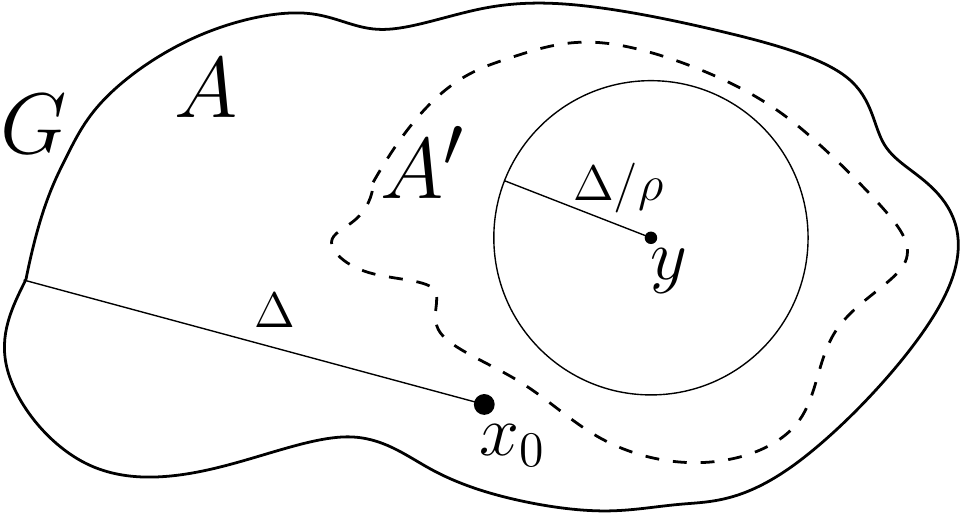}
		\caption{\small  $A$ is a subset of vertices in the graph $G$ with center $x_0$. The radius of $A$ (w.r.t $x_0$) is $\Delta$.
		$A'$ is a subset of $A$ (denoted by the dashed line). As $B(y,\Delta/\rho,G)\subseteq A'$, we say that the vertex $y$ is padded by $A'$.}
		\label{fig:paddedVertex}
	\end{center}
\end{figure}

\subsection{Petal Decomposition}\label{subsec:PetalDecompDesc}

Here we will give a concise description of the Petal decomposition algorithm, focusing on the main properties we will use.
For proofs and further details we refer to \cite{AN12}.

The \texttt{hierarchical-petal-decomposition} (see \algref{alg:h-petal}) is a recursive algorithm. The input is $G[X]$ (a graph $G=(V,E)$ induced over a set of vertices $X\subseteq V$), a center $x_0\in X$, a target $t\in X$, and the radius $\Delta=\Delta_{x_0}(X)$.\footnote{Rather than inferring $\Delta=\Delta_{x_0}(X)$ from $G[X]$ and $x_0$ as in \cite{AN12}, we can think of $\Delta$ as part of the input. We shall allow any $\Delta\ge\Delta_{x_0}(X)$. We stress that in fact in the algorithm we always use $\Delta_{x_0}(X)$, and consider this degree of freedom only in the analysis.}  The algorithm invokes \texttt{petal-decomposition} to partition $X$ into $X_0,X_1,\dots,X_s$ (for some integer $s$), and also provides a set of edges $\{(x_1,y_1),\dots,(x_s,y_s)\}$ and targets $t_0,t_1,\dots,t_s$.
The \texttt{Hierarchical-petal-decomposition} algorithm now recurses on each $(G[X_j],x_j,t_j,\Delta_{x_j}(X_j))$ for $0\le j\le s$, to get trees $\{T_j\}_{0\le j\le s}$, which are then connected by the edges $\{(x_j,y_j)\}_{1\le j\le s}$ to form a spanning tree $T$ for $G[X]$ (the recursion ends when $X_j$ is a singleton). See \figureref{fig:HPDExample} for illustration.
\begin{algorithm}[!ht]
	\caption{$T=\texttt{hierarchical-petal-decomposition}(G[X],x_{0},t,\Delta)$}\label{alg:h-petal}
	\begin{algorithmic}[1]
		\IF {$|X|=1$}
		\RETURN $G[X]$.
		\ENDIF
		\STATE Let $\left(\left\{ X_{j},x_j,t_{j},\Delta_{j}\right\} _{j=0}^{s},\left\{ (y_{j},x_{j})\right\} _{j=1}^{s}\right)=\texttt{petal-decomposition}(G[X],x_{0},t,\Delta)$;
		\FOR {each $j\in[0,\dots,s]$}
		\STATE $T_{j}=\texttt{hierarchical-petal-decomposition}(G[X_{j}],x_{j},t_{j},\Delta_j)$;
		\ENDFOR
		\STATE Let $T$ be the tree formed by connecting $T_{0},\dots,T_{s}$ using the edges $\{y_{1},x_{1}\},\dots,\{y_{s},x_{s}\}$;
	\end{algorithmic}
	
\end{algorithm}

Next we describe the \texttt{petal-decomposition} procedure, see \algref{alg:petal-d}.
Initially it sets $Y_0=X$, and for $j=1,2,\dots,s$ it carves out the petal $X_j$ from the graph induced on $Y_{j-1}$, and sets $Y_j=Y_{j-1}\setminus X_j$ (in order to control the radius increase, the first petal is cut using different parameters). The definition of petal guarantees that $\Delta_{x_0}(Y_j)$ is non-increasing (see \cite[Claim 1]{AN12}), and when at step $s$ it becomes at most $3\Delta/4$, define $X_0=Y_s$ and then the \texttt{petal-decomposition} routine ends. In carving of the petal $X_j\subseteq Y_{j-1}$, the algorithm chooses an arbitrary target $t_j\in Y_{j-1}$ (at distance at least $3\Delta/4$ from $x_0$) and a range $[lo,hi]$ of size $hi-lo\in\{\Delta/8,\Delta/4\}$ which are provided to the sub-routine \texttt{create-petal}.

\begin{algorithm}[!ht]
	\caption{$\left(\left\{ X_{j},x_j,t_{j},\Delta_{j}\right\} _{j=0}^{s},\left\{ (y_{j},x_{j})\right\} _{j=1}^{s}\right)=\texttt{petal-decomposition}(G[X],
		x_{0},t,\Delta)$}\label{alg:petal-d}
	\begin{algorithmic}[1]
		\STATE 
		Let $Y_0=X$;
		Set $j=1$;
		\IF {$d_{X}(x_{0},t)\ge\Delta/2$}
		\STATE Let
		$X_1=\texttt{create-petal}(G[Y_0],[d_{X}(x_{0},t)-\Delta/2,d_{X}(x_{0},t)-\Delta/4],x_0,t)$;
		\STATE $Y_{1}=Y_{0}\setminus X_{1}$;
		\STATE Let $\{x_1,y_{1}\}$ be the unique edge on the shortest path $P_{x_{0}t}$ from $x_0$ to $t$ in $Y_0$, where $x_1\in X_1$ and $y_1\in Y_1$.
		\STATE Set $t_{0}=y_{1}$, $t_{1}=t$; $j=2$;
		\ELSE
		\STATE set $t_{0}=t$.
		\ENDIF
		\WHILE {$Y_{j-1}\setminus B_{X}(x_{0},\frac{3}{4}\Delta)\neq\emptyset$}
		\STATE Let $t_{j}\in Y_{j-1}$ be an arbitrary vertex satisfying $d_{X}(x_{0},t_{j})>\frac{3}{4}\Delta$;
		\STATE Let
		$X_j=\texttt{create-petal}(G[Y_j],[0,\Delta/8],x_0,t_j)$;
		\STATE $Y_{j}=Y_{j-1}\setminus X_{j}$;
		\STATE Let $\{x_j,y_{j}\}$ be the unique edge on the shortest path $P_{x_{j}t_j}$ from $x_0$ to $t_j$ in $Y_{j-1}$, where $x_j\in X_j$ and $y_j\in Y_j$.
		\STATE For each edge $e\in P_{x_{j}t_{j}}$, set its weight to be $w(e)/2$; \label{line:edgeWeightReduce}
		\STATE Let $j=j+1;$		
		
		\ENDWHILE
		
		\STATE Let $s=j-1$;
		\STATE Let $X_{0}=Y_{s}$;
		\RETURN $\left(\left\{ X_{j},x_{j},t_{j},\Delta_{x_j}(X_j)\right\} _{j=0}^{s},\left\{ (y_{j},x_{j})\right\} _{j=1}^{s}\right)$;
	\end{algorithmic}
\end{algorithm}
(One may notice that in \lineref{line:edgeWeightReduce} of the \texttt{petal-decomposition} procedure,
the weight of some edges is changed by a factor of 2. This can happen at most once for every edge throughout the  \texttt{hierarchical-petal-decomposition} execution, thus it may affect the padding parameter by a factor of at most 2. This re-weighting is ignored for simplicity.)

Both \texttt{hierarchical-petal-decomposition} and \texttt{petal-decomposition} are essentially the algorithms that appeared in \cite{AN12}. The main difference from their work lies in the \texttt{create-petal} procedure, depicted in \algref{alg-pick-rad}. It carefully selects a radius $r\in[lo,hi]$, which determines the petal $X_j$ together with a connecting edge $(x_j,y_j)\in E$, where $x_j\in X_j$ is the center of $X_j$ and $y_j\in Y_j$. It is important to note that the target $t_0\in X_0$ of the central cluster $X_0$ is determined during the creation of the first petal $X_1$. It uses an alternative metric on the graph, known as {\em cone-metric}:
\begin{definition}[Cone-metric] Given a graph $G=(V,E)$, a subset $X\subseteq V$
	and points $x,y\in X$, define the $\emph{cone-metric}$ $\rho=\rho(X,x,y):X^{2}\to\mathbb{R}^{+}$
	as $\rho(u,v)=\left|\left(d_{X}(x,u)-d_{X}(y,u)\right)-\left(d_{X}(x,v)-d_{X}(y,v)\right)\right|$.
\end{definition}
(The cone-metric is in fact a pseudo-metric, i.e., distances between distinct points are allowed to be 0.) 
The ball $B_{(X,\rho)}(y,r)$ in the cone-metric
$\rho=\rho(X,x,y)$, contains all vertices $u$ whose shortest path to $x$ is increased (additively) by at most $r$ if forced to go through $y$.

In the $\texttt{create-petal}$
algorithm, while working in a subgraph $G[Y]$ with two specified vertices: a center $x_{0}$
and a target $t$, we define $W_{r}\left(Y,x_{0},t\right)=\bigcup_{p\in P_{x_{0}t}:\ d_{Y}(p,t)\le r}B_{(Y,\rho(Y,x_{0},p))}(p,\frac{r-d_{Y}(p,t)}{2})$
which is union of balls in the cone-metric, where any vertex $p$
in the shortest path from $x_{0}$ to $t$ of distance at most $r$
from $t$ is a center of a ball with radius $\frac{r-d_{Y}(p,t)}{2}$.

\begin{figure}
	\begin{center}
		\includegraphics[width=0.52\textwidth]{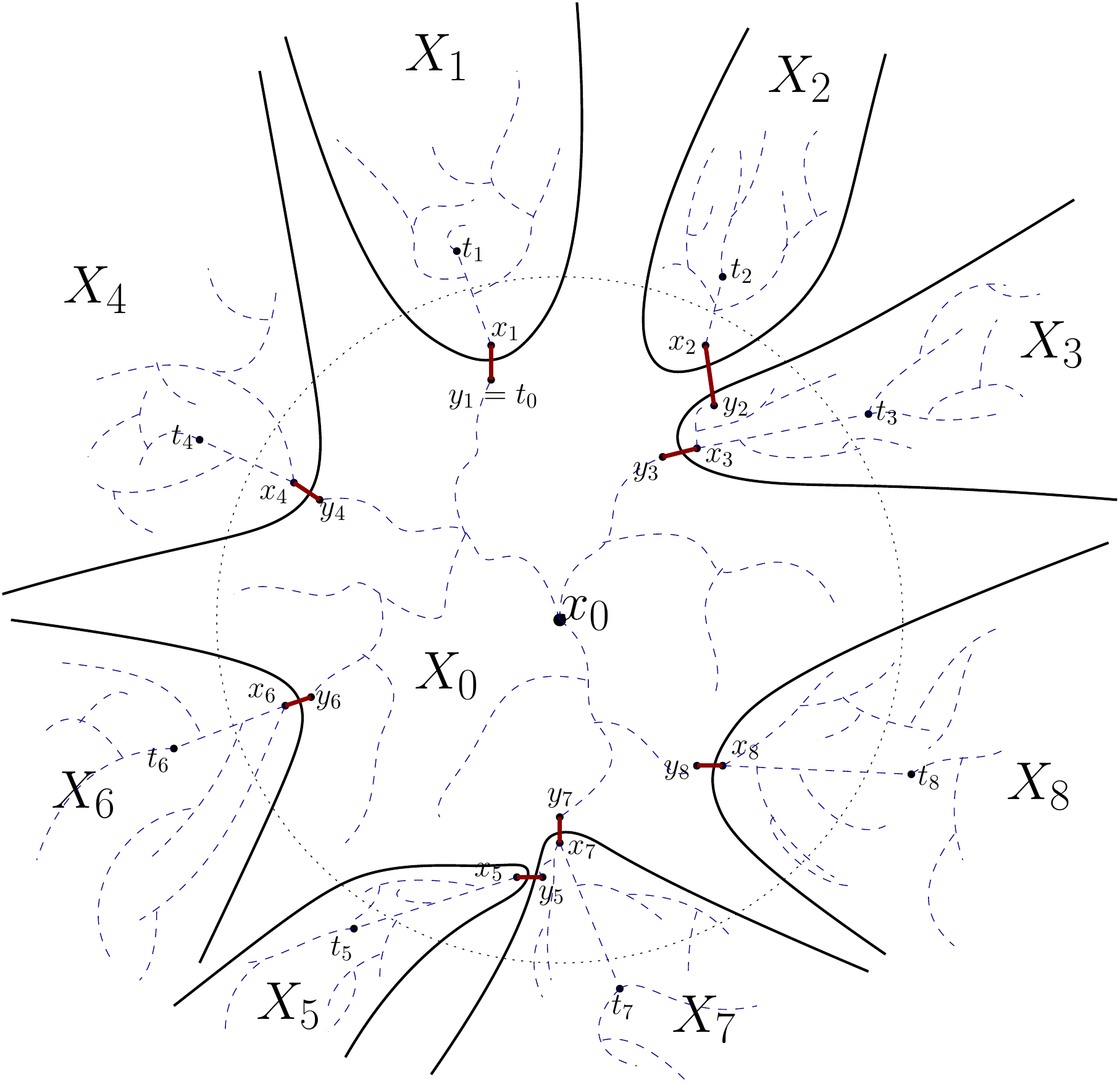}
		\caption{A schematic depiction of a partition done by the \texttt{petal-decomposition} algorithm. The doted ball contain the vertices within $3\Delta/4$ from the center $x_0$. The algorithm iteratively picks a target $t_j$ outside this ball, and builds a petal $X_j$, that will be connected to the rest of the graph by the edge $(x_j,y_j)$. The doted lines represent the trees created recursively in each $X_j$.}
		\label{fig:HPDExample}
	\end{center}
\end{figure}

The following facts are from \cite{AN12}.
\begin{fact}
	\label{FactPetalTreeRadiusBound4} Running \texttt{Hierarchical-petal-decomposition} on input $(G[X],x_0,t,\Delta_{x_0}(X))$ will provide a spanning tree $T$ satisfying
	\[
	\Delta_{x_0}(T)\le 4\Delta_{x_0}(X).
	\]
\end{fact}
\begin{fact}
	\label{FactPetal3/4Radius} If the \texttt{petal-decomposition} partitions $X$ with center $x_0$ into $X_0,\dots,X_s$ with centers $x_0,\dots,x_s$, then for any $0\le j\le s$ we have $\Delta_{x_j}(X_j)\le (3/4)\cdot\Delta_{x_0}(X)$.
\end{fact}

We will need the following observation. Roughly speaking, it says that when the \texttt{petal-decomposition} algorithm is carving out $X_{j+1}$, it is oblivious to the past petals $X_1,\dots,X_j$, edges and targets -- it only cares about $Y_j$ and the original diameter $\Delta$.
\begin{observation}\label{ob:delta}
Assume that \texttt{petal-decomposition} on input $(G\left[X\right],x_{0},t,\Delta_{x_0}(X))$ returns as output $(X_{0},X_{1},\dots,X_{s},\left\{ y_{1},x_{1}\right\} ,\dots,\left\{ y_{s},x_{s}\right\} ,t_{0},\dots,t_{s})$. Then running \texttt{petal-decomposition} on input $(G\left[Y_j\right],x_{0},t_0,\Delta_{x_0}(X))$ will output $(X_{0},X_{j+1}\dots,X_{s},\left\{ y_{j+1},x_{j+1}\right\} ,\dots,\left\{ y_{s},x_{s}\right\} ,t_{0},t_{j+1}\dots,t_{s})$.
\end{observation}

\subsection{Choosing a Radius}\label{subsec:createPetal}

Fix some $1\le j\le s$, and consider carving the petal $X_j$ from the graph induced on $Y=Y_{j-1}$. While the algorithm of \cite{AN12} described a specific way to choose the radius, we require a somewhat more refined choice. The properties of the petal decomposition described above (in \subsectionref{subsec:PetalDecompDesc}) together with \factref{FactPetal3/4Radius} and \factref{FactPetalTreeRadiusBound4},  hold for any radius picked from a given interval. We will now describe the method to select a radius that suits our needs. The \texttt{petal-decomposition} algorithm provides an interval $[lo,hi]$ of size at least $\Delta/8$, and for each $r\in[lo,hi]$ let $W_{r}(Y,x_0,t)\subseteq Y$ denote the petal of radius $r$ (usually we will omit $(Y,x_0,t)$).
The following fact demonstrates that petals are similar to balls.

\begin{fact}\label{FctW_rProp}
For every $y\in W_{r}$ and $l\ge 0$, $B(y,l,Y)\subseteq W_{r+4l}$.
\end{fact}
Note that \factref{FctW_rProp} implies that $W_r$ is monotone in $r$, i.e., for $r\le r'$ it holds that $W_{r}\subseteq W_{r'}$.

Our algorithm will maintain a set of {\em marked} vertices $M\subseteq V$, and will update it in any petal creation. Roughly speaking, the marked vertices are those that are fully padded in the (partial) hierarchical partition generated so far by the algorithm. If initially $|M|=m$, we want that at the end of the process at least $m^{1-1/k}$ vertices will remain marked.
In the partition of $X$ to $X_0,\dots,X_s$, some of the marked vertices will be $\rho$-padded by a petal $X_j$ (w.r.t. $X$), and some of the others will be unmarked, by the following rule.  \factref{FctW_rProp} implies that if we choose a radius $r$ when creating some petal $X_j=W_r$, then all marked vertices in $W_{r-4\Delta/\rho}$ will be $\rho$-padded by $X_j$, and thus remain marked. All the marked vertices in $W_{r+4\Delta/\rho}\setminus W_{r-4\Delta/\rho}$ are considered unmarked from now on, since their $\Delta/\rho$ ball may intersect more than one cluster in the current partition (note that some of these vertices can be outside $X_j$).

Our algorithm to select a radius is based on region growing techniques, similar to those in \algref{alg-strong-ramsey}, but rather more involved.
Since in the petal decomposition framework we cannot pick as center a vertex maximizing the "small ball", we first choose an appropriate range that mimics that choice (see e.g. \lineref{line:a-b} in the algorithm below) -- this is the reason for the extra factor of $\log\log n$.
The basic idea in region growing is to charge the number of marked vertices whose ball is cut by the partition (those in $W_{r+4\Delta/\rho}\setminus W_{r-4\Delta/\rho}$), to those that are saved (in $W_{r-4\Delta/\rho}$). In our setting we are very sensitive to constant factors in this charging scheme (as opposed to the average stretch considered in \cite{EEST05}), because these constants are multiplied throughout the recursion. In particular, we must avoid a range in $[lo,hi]$ that contains more than half of the marked vertices, a constraint which did not exist in previous manifestation of this region growing scheme. To this end, if the first half $[lo,mid]$ (with $mid=(hi+lo)/2$) is not suitable, we must "cut backwards" in the regime $[mid,hi]$, and charge the marked vertices that were removed from $M$ to the remaining graph $Y_j$, rather than to those saved in the created cluster $X_j$.\footnote{We remark that paying a factor of $\log\log(nW)$, where $W$ is the maximum edge weight, might have simplified the algorithm slightly.}

\begin{algorithm}[H]
	\caption{$X=\texttt{create-petal}(G[Y],[lo,hi],x_0,t)$}\label{alg-pick-rad}
	\begin{algorithmic}[1]
		\STATE Let $m=|Y\cap M|$;
		\STATE $L=\lceil1+\log\log m\rceil$;
		\STATE $R=hi-lo$; $mid=(lo+hi)/2=lo+R/2$;
		\STATE For every $r$, denote $W_{r}=W_{r}(Y,x_0,t)$,   $w_r=|M\cap W_r|$;		
		\IF {$w_{mid}\le \frac{m}{2}$}
		\IF {$w_{lo+\frac{R}{2L}}=0$}
		\STATE Set $r=lo+\frac{R}{4L}$;
		\ELSE
		\STATE  Choose $\left[a,b\right]\subseteq\left[lo,mid\right]$\label{line:a-b}
		such that $b-a=\frac{R}{2L}$ and $w_a\ge w_b^2/m$; \COMMENT{see \lemmaref{lem:interval-choose}}
		
		\STATE Pick $r\in\left[a,b\right]$ such that
		$w_{r+\frac{b-a}{2k}}\le w_{r-\frac{b-a}{2k}}\cdot\left(\frac{w_{b}}{w_{a}}\right)^{\frac{1}{k}}$; \COMMENT{see \lemmaref{lem:radiuosRG}}
		\ENDIF
		\ELSE 
		\STATE For every $r\in[lo,hi]$, denote  $q_r=|(Y\setminus W_r)\cap M|$;
		\IF {$q_{hi-\frac{R}{2L}}=0$}
		\STATE Set $r=hi-\frac{R}{4L}$;
		\ELSE
		\STATE  Choose $\left[b,a\right]\subseteq\left[mid,hi\right]$
		such that $a-b=\frac{R}{2L}$ and $q_a\ge q_b^2/m$; \COMMENT{see \lemmaref{lem:interval-choose-Ver2}}
		
		\STATE Pick $r\in\left[b,a\right]$ such that
		$q_{r-\frac{a-b}{2k}}\le q_{r+\frac{a-b}{2k}}\cdot\left(\frac{q_{b}}{q_{a}}\right)^{\frac{1}{k}}$; \COMMENT{see \lemmaref{lem:radiuosRG-Ver2}}	
		\ENDIF
		\ENDIF	
		\STATE $M\leftarrow M\setminus(W_{r+\frac{R}{4Lk}}\setminus W_{r-\frac{R}{4Lk}})$\label{line:M}
		\RETURN $W_r$;
	\end{algorithmic}
\end{algorithm}

\subsection{Proof of Correctness}
Let $z\in V$ be an arbitrary vertex, given a set $M\subseteq V$, let $T$ be the tree returned by calling \texttt{Hierarchical-petal-decomposition} on $(G[V],z,z,\Delta_z(V))$ and marked vertices $M$. There is a natural Strong Diameter Hierarchical Partition (\textbf{SDHP}) ${\cal X}=\{\mathcal{X}_i\}_{i=1}^{\Phi}$ associated with the tree $T$, where
$\mathcal{X}_i$ consists of all the clusters created in level $\Phi-i$ of the recursion (and $\mathcal{X}_\Phi=\{V\}$).
By \factref{FactPetal3/4Radius}, the radius is always non-increasing. Hence ${\cal X}$ is indeed a \textbf{SDHP}. Denote by $\sur(M)\subseteq M$ the set of vertices that remained marked throughout the execution of \texttt{Hierarchical-petal-decomposition}.

\begin{lemma}\label{lem:active-padded}
Suppose $x\in \sur(M)$, then $x$ is $\rho$-fully-padded in ${\cal X}$.
\end{lemma}
\begin{proof}
Fix any $2\le i\le \Phi$. Let $X\in{\cal X}_i$ be the cluster containing $x$ in the $(\Phi-i)$-th level of the recursion with $\Delta=\Delta(X)$. Assume $X$ was partitioned by \texttt{petal-decomposition} into $X_0,\dots,X_s$, and let $X_j\subseteq X$ be the cluster containing $x\in X_j$. Assuming (inductively) that $x$ was $\rho$-padded by $X$, we need to show that it is also $\rho$-padded by $X_j$, that is, $B=B(x,\Delta/\rho,G)\subseteq X_j$. (Note that $B\subseteq X$ since the radii are non-increasing, so $x$ is padded in all higher levels.)

First we argue that none of the petals $X_1,\dots,X_{j-1}$ intersects $B$. Seeking contradiction, assume it is not the case, and let $1\le j'<j$ be the minimal such that there exists $y\in X_{j'}\cap B$. By the minimality of $j'$ it follows that $B\subseteq Y'=Y_{j'-1}$, and thus $\dist(x,y,Y')=\dist(x,y,G)\le \Delta/\rho$. Let $r'$ be the radius chosen when creating the petal $X_{j'}=W'_{r'}$, and \factref{FctW_rProp} implies that
\[
x\in B(y,\Delta/\rho,Y')\subseteq W'_{r'+4\Delta/\rho}=W'_{r'+R/(4Lk)}~,
\]
where we recall that $\Delta=8R$ and $\rho=2^7Lk$. This is a contradiction to the fact that $x\in \sur(X)$: clearly $x\notin W'_{r'-R/(4Lk)}$ since it is not included in $X_{j'}=W'_{r'}$ (and using the monotonicity of $W'_r$), so it should have been removed from $M$ when creating $X_{j'}$ (in \lineref{line:M} of the algorithm).

For the case $j=0$ the same reasoning shows $B$ does not intersect any petal $X_1,\dots ,X_s$ and we are done. For $j>0$, it remains to show that $B\subseteq X_j$, but this follows by a similar calculation. Let $r$ be the radius chosen for creating the petal $X_j=W_r$, and $Y=Y_{j-1}$. We have $B\subseteq Y$, and since $x\in \sur(X)$ it must be that $x\in W_{r-R/(4Lk)}$. Again by \factref{FctW_rProp} we have
\[
B=B(x,\Delta/\rho,G)=B(x,\Delta/\rho,Y)\subseteq W_{r-R/(4Lk)+4\Delta/\rho} = W_r=X_j~.\qedhere
\]
\end{proof}

\begin{lemma}\label{lem:paddedImpliesStretch}
	Consider a vertex $v\in \sur(M)$, then for every $u\in V$, $\dist(v,u,T)\le 8\rho\cdot \dist(v,u,G)$.
\end{lemma}
\begin{proof}
Let ${\cal X}=\{\mathcal{X}_i\}_{i=1}^{\Phi}$ be the \textbf{SDHP} associated with $T$, and for $1\le i\le \Phi$ let $A_i\in{\cal X}_i$ be the cluster containing $v$. Take the minimal $2\le i\le \Phi$ such that $u\in A_i$ (there exists such an $i$ since $u\in A_\Phi=V$ and $u\notin A_1=\{v\}$). By \lemmaref{lem:active-padded} $v$ is $\rho$-fully-padded, so we have that $B(v,\Delta/\rho,G)\subseteq A_{i-1}$, where $\Delta=\Delta(A_i)$. But as $u\notin A_{i-1}$, it must be that $\dist(u,v,G)>\Delta/\rho$. Since both $u,v\in A_i$, \factref{FactPetalTreeRadiusBound4} implies that the radius of the tree created for $A_i$ is at most $4\Delta$, so that
	\[
	\dist(u,v,T)\le 2\cdot 4\Delta\le 8\rho\cdot \dist(u,v,G)~.\qedhere
	\]
\end{proof}

\begin{lemma}\label{lem:SavedTerminals}
$|\sur(M)|\ge |M|^{1-1/k}$.
\end{lemma}
\begin{proof}

We prove by induction on $|X|$ that if a cluster $X\in{\cal X}_i$ (for some $1\le i\le \Phi$) has currently $m$ marked vertices, then at the end of the process at least $m^{1-1/k}$ of them will remain marked.

\sloppy The base case when $X$ is a singleton is trivial. For the inductive step, assume we call \texttt{petal-decomposition} on $(G[X],x_0,t,\Delta)$ with $\Delta\ge\Delta_{x_0}(X)$ and the current marked vertices $\hat{M}$. Assume that the \texttt{petal-decomposition} algorithm does a non-trivial partition of $X$ to $X_0,\dots,X_s$ (if it is the case that all vertices are sufficiently close to $x_0$, then no petals will be created, and the\\ \texttt{hierarchical-petal-decomposition} will simply recurse on $(G[X],x_0,t,\Delta_{x_0}(X))$, so we can ignore this case).  Denote by $M_j$ the marked vertices that remain in $X_j$ (just before the recursive call on $X_j$), and recall that $\sur(X_j)$ is the set of vertices of $M_j$ that remain marked until the end of the \texttt{hierarchical-petal-decomposition} algorithm. Then $\sur(X)=\bigcup_{0\le j\le s}\sur(X_j)$, and we want to prove that $|\sur(X)|\ge m^{1-1/k}$.

Let $X_1=W_r$ be the first petal created by the \texttt{petal-decomposition} algorithm, and $Y_1=X\setminus X_1$.
Denote by $\res(X_1)=W_{r+R/(4Lk)}\cap \hat{M}$ the responsibility set for $X_1$ (i.e. the marked vertices that are either in $M_1$ or were removed from $\hat{M}$ when $X_1$ was created). Define $M'=\hat{M}\setminus \res(X_1)$, the set of marked vertices that remain in $Y_1$.
By \obsref{ob:delta}, we can consider the remaining execution of \texttt{petal-decomposition} on $Y_1$ as a new recursive call of \texttt{petal-decomposition} with input $(G[Y_1],x_0,t_0,\Delta)$ and marked vertices $M'$. Since $|X_1|,|Y_1|<|X|$, the induction hypothesis implies that $|\sur(X_1)|\ge|M_1|^{1-1/k}$ and $|\sur(Y_1)|\ge|M'|^{1-1/k}$.

We now do a case analysis according to the choice of radius in \algref{alg-pick-rad}.
\begin{enumerate}
\item {\bf Case 1:} $w_{mid}\le m/2 $ and $w_{lo+R/(2L)}=0$. In this case we set $r=lo+R/(4L)$. Note that $w_{r+R/(4Lk)}\le w_{lo+R/(2L)}=0$, so $M'=\hat{M}$, and by the induction hypothesis on $Y_1$, the number of fully padded vertices is $|\sur(X)|= |\sur(Y_1)|\ge|\hat{M}|^{1-1/k}=m^{1-1/k}$, as required.

\item {\bf Case 2:} $w_{mid}\le m/2 $ and $w_{lo+R/(2L)}>0$. In this case we pick $a,b\in[lo,hi]$ so that $b-a=R/(2L)$ and
    \begin{equation}\label{eq:wawb}
    w_a>w_b^2/m~,
    \end{equation} and also choose $r\in[a,b]$ such that $w_{r+\frac{b-a}{2k}}\le w_{r-\frac{b-a}{2k}}\cdot\left(\frac{w_{b}}{w_{a}}\right)^{1/k}$. As $\frac{b-a}{2k}=\frac{R}{4Lk}$ and $|M_1|=w_{r-R/(4Lk)}$ we have that
    \begin{equation}\label{eq:KC}
    |M_1|\ge\res(X_1)\cdot\left(\frac{w_{a}}{w_{b}}\right)^{1/k}~.
    \end{equation}
    By the induction hypothesis on $X_1$ we have that
	\begin{align*}
	|\sur(X_{1})| & \ge\frac{|M_1|}{|M_1|^{1/k}}\stackrel{\eqref{eq:KC}}{\ge}|\res(X_{1})|\cdot\left(\frac{w_{a}}{|M_1|\cdot w_{b}}\right)^{1/k}\\
	& \stackrel{\eqref{eq:wawb}}{\ge}|\res(X_{1})|\cdot\left(\frac{w_{b}}{m\cdot|M_1|}\right)^{1/k}\ge\frac{|\res(X_{1})|}{m^{1/k}}~,
	\end{align*}
    where in the last inequality we use that $|M_1|= w_{r-(b-a)/(2k)}\le w_b$. Now by the induction hypothesis on $Y_1$ we get
    \begin{align*}
    |\sur(X)| & =|\sur(Y_{1})|+|\sur(X_{1})|\\
    & \ge|M'|^{1-1/k}+\frac{|\res(X_{1})|}{m^{1/k}}\ge\frac{|M'|+|\res(X_{1})|}{m^{1/k}}=\frac{|\hat{M}|}{m^{1/k}}=m^{1-1/k}
    \end{align*}

\item {\bf Case 3:} $w_{mid}> m/2 $ and $q_{hi-R/(2L)}=0$. In this case we set $r=hi-R/(4L)$. Note that $q_{r-R/(4Lk)}\le q_{hi-R/(2L)}=0$ (recall that $q_r$ is non-increasing in $r$, by \factref{FctW_rProp}), so $M_1=\hat{M}$, and by the induction hypothesis on $X_1$, $|\sur(X)|=|\sur(X_1)|\ge|M_1|^{1-1/k}=m^{1-1/k}$, as required.

\item {\bf Case 4:} $w_{mid}> m/2 $ and $q_{hi-R/(2L)}>0$. In this case we pick $a,b\in[lo,hi]$ so that $a-b=R/(2L)$ and
    \begin{equation}\label{eq:qaqb}
    q_a>q_b^2/m~,
    \end{equation}
    and also choose $r\in[b,a]$ such that $q_{r-\frac{b-a}{2k}}\le q_{r+\frac{b-a}{2k}}\cdot\left(\frac{q_{b}}{q_{a}}\right)^{1/k}$. In this case when we cut "backwards", we shift the responsibility for the vertices unmarked by the creation of $X_1$ to $Y_1$. This is captured by defining $\res(Y_1)=\hat{M}\setminus M_1$.
    Since $|M'|=q_{r+\frac{a-b}{2k}}$ and $|\res(Y_1)|= q_{r-\frac{a-b}{2k}}$ we have
    \begin{equation}\label{eq:KCq}
    |M'|\ge|\res(Y_1)|\cdot\left(\frac{q_{a}}{q_{b}}\right)^{1/k}~.
    \end{equation}
    By the induction hypothesis on $Y_1$ we have that
	\begin{align*}
	|\sur(Y_{1})| & \ge\frac{|M'|}{|M'|^{1/k}}\stackrel{\eqref{eq:KCq}}{\ge}|\res(Y_{1})|\cdot\left(\frac{q_{a}}{|M'|\cdot q_{b}}\right)^{1/k}\\
	& \stackrel{\eqref{eq:qaqb}}{\ge}|\res(Y_{1})|\cdot\left(\frac{q_{b}}{m\cdot|M'|}\right)^{1/k}\ge\frac{|\res(Y_{1})|}{m^{1/k}}~,
	\end{align*}
    where in the last inequality we use that $|M'|= q_{r+(a-b)/(2k)}\le q_b$. Now by the induction hypothesis on $X_1$ we get
    \begin{align*}
    |\sur(X)| & =|\sur(Y_{1})|+|\sur(X_{1})|\\
    & \ge\frac{|\res(Y_{1})|}{m^{1/k}}+|M_1|^{1-1/k}\ge\frac{|\res(Y_{1})|+|M_1|}{m^{1/k}}=\frac{|\hat{M}|}{m^{1/k}}=m^{1-1/k}~.
    \end{align*}
\end{enumerate}

\end{proof}
From \lemmaref{lem:paddedImpliesStretch} and \lemmaref{lem:SavedTerminals} we derive the following theorem.
\begin{theorem}\label{thm:main}
Let $G=(V,E)$ be a weighted graph, fix a set $M\subseteq V$ of size $m$ and a parameter $k\ge 1$. There exists a spanning tree $T$ of $G$, and a set $\sur(M)\subseteq M$ of size at least $m^{1-1/k}$, such that for every $v\in \sur(M)$ and every $u\in V$ it holds that $\dist(v,u,T)\le O(k\log\log m)\cdot \dist(v,u,G)$.
\end{theorem}

We conclude with the proof of our main result.
\begin{proof}[Proof of \theoremref{thm:main-col}]
Set $M_1=V$, and for $i\ge 1$ define $M_{i+1}=M_i\setminus\sur(M_i)$. We shall apply \theoremref{thm:main} iteratively, where $M_i$ is the set of vertices given as input to the $i$-th iteration, that has size $|M_i|=m_i$. Let $T_i$ be the tree created in iteration $i$. By \theoremref{thm:main} the sizes $m_1,m_2,\dots$ obey the recurrence $m_1=n$ and $m_{i+1}\le m_i-m_i^{1-1/k}$, which implies that after $k\cdot n^{1/k}$ iterations we will have $m_{k\cdot n^{1/k}+1}<1$ (see \cite[Lemma 4.2]{MN07}), and thus every vertex is in $\sur(M_i)$ for some $1\le i\le k\cdot n^{1/k}$. For each $v\in V$, let $\home(v)$ be the tree $T_i$ such that $v\in \sur(M_i)$.
\end{proof}

\subsection{Routing with Short Labels}\label{sec:routing}

In this section we prove \theoremref{thm:route}.
We first use a result of \cite{TZ01b} concerning routing in trees.
\begin{theorem}[\cite{TZ01b}]\label{thm:tree-routh}
	For any tree $T=(V,E)$ (where $|V|=n$), there is a routing scheme with stretch 1 that has routing tables of size $O(b)$ and labels of size $(1+o(1))\log_bn$. The decision time in each vertex is $O(1)$.
\end{theorem}

Combining \theoremref{thm:main-col} and \theoremref{thm:tree-routh} we can construct a routing scheme. Let $\mathcal{T}$ be the set of trees from \theoremref{thm:main-col}.
Each tree $T\in\mathcal{T}$ is associated with a routing scheme given by \theoremref{thm:tree-routh}. Set $L_T(x)$ be the label of the vertex $x$ in the routing scheme of the tree $T$.

In our scheme, the routing table of each vertex will be the collection of its routing tables in all the trees in ${\cal T}$. Hence the table size is $O(b)\cdot |\mathcal{T}|=O(k\cdot b\cdot n^{1/k})$.
The label of each $x\in V$ will be   $(\home(x),L_{\home(x)}(x))$, i.e., the name of the home tree of $x$ and the label of $x$ in that tree. The label size is $1+(1+o(1))\log_bn=(1+o(1))\log_bn$.

The routing is done in a straightforward manner, to route from $y$ to $x$, we extract $\home(x)$ from the given label of $x$, and simply use the routing scheme of the tree $\home(x)$. Note that this process takes $O(1)$ time, and is independent of the routing path traversed so far. Since all vertices store in their routing table the appropriate routing information for $\home(x)$, the routing can be completed.

\bibliographystyle{alpha}
\bibliography{bib-extended,art}

\appendix
\section{Proof of Correctness for \algref{alg-pick-rad}}
In this section we prove that the choices made in the \texttt{create-petal} procedure are all legal. In all the Lemmas that follow, we shall use the notation in \algref{alg-pick-rad}.
\begin{lemma}\label{lem:interval-choose}
	If $w_{mid}\le\frac{m}{2}$ and $w_{lo+\frac{R}{2L}}\ge1$,
	then
	there is $\left[a,b\right]\subseteq\left[lo,mid\right]$
	such that $b-a=\frac{R}{2L}$ and $w_a\ge w_b^2/m$.
\end{lemma}
\begin{proof}
	Seeking contradiction, assume that for every such $a,b$ with $b-a=\frac{R}{2L}$ it holds that $w_b>\sqrt{m\cdot w_a}$. Applying this on $b=mid-\frac{i R}{2L}$ and $a=mid-\frac{(i+1) R}{2L}$ for every $i=0,1,\dots,L-2$, we have that
	\begin{align*}
	w_{mid} & >m^{1/2}\cdot w_{mid-\frac{R}{2L}}^{1/2}>\dots
	 >m^{1-2^{-(L-1)}}\cdot w_{mid-\frac{(L-1)R}{2L}}^{2^{-(L-1)}}\ge m\cdot2^{-1}\cdot w_{lo+\frac{R}{2L}}^{1/(2\log m)}\ge\frac{m}{2}~,
	\end{align*}
	where we used that $1+\log\log m\le L\le 2+\log\log m$ and $mid=lo+R/2$. In the last inequality we also used that $w_a\ge 1$, which follows since $b=a+\frac{R}{2L}\ge lo+\frac{R}{2L}$, thus $w_b\ge 1$, and in particular $w_{a}\ge w_{b}^{2}/m>0$. The contradiction follows.
\end{proof}
\begin{lemma}\label{lem:radiuosRG}
	There is $r\in\left[a,b\right]$ such that $w_{r+\frac{b-a}{2k}}\le w_{r-\frac{b-a}{2k}}\cdot\left(\frac{w_{b}}{w_{a}}\right)^{\frac{1}{k}}$.
\end{lemma}
\begin{proof}
	Seeking contradiction, assume there is no such choice of $r$, then applying this for $r=b-(i+1/2)\cdot\frac{b-a}{k}$ for $i=0,1,\dots,k-1$ we get
	\[
		 w_{b}>w_{b-\frac{b-a}{k}}\cdot\left(\frac{w_{b}}{w_{a}}\right)^{1/k}>\cdots>w_{b-k\cdot\frac{b-a}{k}}\cdot\left(\frac{w_{b}}{w_{a}}\right)^{k/k}=w_{a}\cdot\frac{w_{b}}{w_{a}}=w_{b}~,
	\]
	a contradiction.
\end{proof}
The following two lemmas are symmetric to the two lemmas above.
\begin{lemma}\label{lem:interval-choose-Ver2}
	If $w_{mid}>\frac{m}{2}$ (implies $q_{mid}\le\frac{m}{2}$) and $q_{hi-\frac{R}{2L}}\ge1$,	
	then
	there is $\left[b,a\right]\subseteq\left[mid,hi\right]$
	such that $a-b=\frac{R}{2L}$ and $q_a\ge q_b^2/m$.
\end{lemma}
\begin{lemma}\label{lem:radiuosRG-Ver2}
	There is $r\in\left[b,a\right]$ such that
	$q_{r-\frac{a-b}{2k}}\le q_{r+\frac{a-b}{2k}}\cdot\left(\frac{q_{b}}{q_{a}}\right)^{1/k}$.
\end{lemma}
\section{Table of Distance oracles}\label{sec:table}
\begin{table}[H]
	\centering
	\label{tab:DistanceOracle}
	\begin{tabular}{|l|l|l|l|l|}
\hline
\textbf{Distance Oracle}& \textbf{Stretch} 		& \textbf{Size}        	& \textbf{Query time} & \textbf{Is deterministic?} \\ \hline
\cite{TZ01}             & $2k-1$           		& $O(k\cdot n^{1+1/k})$ & $O(k)$             & no                         \\ \hline
\cite{MN07}      		& $128k$           		& $O(n^{1+1/k})$        & $O(1)$              & no                         \\ \hline
\cite{W13}         		& $(2+\epsilon)k$  		& $O(k\cdot n^{1+1/k})$ & $O(1/\epsilon)$     & no                         \\ \hline
\cite{C14}         		& $2k-1$           		& $O(k\cdot n^{1+1/k})$ & $O(1)$              & no                         \\ \hline
\cite{C15}         		& $2k-1$           		& $O(n^{1+1/k})$        & $O(1)$              & no                         \\ \noalign{\global\arrayrulewidth0.035cm}
 \hline
 \noalign{\global\arrayrulewidth0.4pt}
\cite{RTZ05}       		& $2k-1$           		& $O(k\cdot n^{1+1/k})$ & $O(k)$             & yes                        \\ \hline
\cite{W13}    		& $2k-1$           		& $O(k\cdot n^{1+1/k})$ & $O(\log k)$         & yes                       \\  \hline
\textbf{This paper}     & $8 (1+\eps)k$  		& $O(n^{1+1/k})$        & $O(1/\epsilon)$     & yes            \\ \hline
\textbf{This paper}     & $2k-1$				& $O(k\cdot n^{1+1/k})$ & $O(1)$              & yes                        \\ \hline
	\end{tabular}
	\caption{Different distance oracles }
\end{table}
\end{document}